\documentclass[a4paper,fleqn]{article}
\usepackage{a4wide}
\usepackage{amsmath,amssymb,amsthm}
\usepackage{thmtools}
\usepackage{enumerate}
\usepackage{graphicx,tikz}
\usetikzlibrary{arrows,calc}
\usepackage{abstract}
\usepackage{authblk}

\usepackage{hyperref}
\usepackage[T1]{fontenc}
\usepackage{mathpazo}
\makeatletter\@ifpackageloaded{mathpazo}\@tempswatrue\@tempswafalse
\if@tempswa
  \DeclareFontFamily{OT1}{pzc}{}
  \DeclareFontShape{OT1}{pzc}{m}{it}{<-> s * [1.15] pzcmi7t}{}
  \DeclareMathAlphabet{\mathpzc}{OT1}{pzc}{m}{it}
\fi\makeatother
\usepackage{bm}
\usepackage[utf8]{inputenc}

\usepackage[sorting=nyt,giveninits=true,maxbibnames=5,backend=biber]{biblatex}
\makeatletter\@ifpackageloaded{biblatex}{%
  \usepackage{csquotes} 
  \bibliography{../../references}
  \renewbibmacro{in:}{%
    \ifentrytype{incollection}{\printtext{\bibstring{in}\intitlepunct}}{}}
  \renewbibmacro{publisher+location+date}{%
    \iflistundef{publisher}
      {\setunit*{\addcomma\space}}
      {\setunit*{\addcomma\space}}%
    \printlist{publisher}%
    \setunit*{\addcomma\space}%
    \printlist{location}%
    \setunit*{\addcomma\space}%
    \usebibmacro{date}%
    \newunit}
  \DeclareFieldFormat[article]{pages}{#1\isdot}
  \DeclareFieldFormat[article,incollection,inproceedings,unpublished,eprint]{title}{#1\isdot}
  \DeclareFieldFormat[thesis]{title}{\mkbibemph{#1\isdot}}
  \DeclareFieldFormat[unpublished]{date}{(#1)\isdot}
  \DeclareFieldFormat[unpublished]{note}{#1\nopunct} 
  \DeclareFieldFormat[eprint]{date}{(#1)\isdot}
  \DeclareFieldFormat[eprint]{note}{#1\nopunct} 
  \DeclareFieldFormat[article]{journaltitle}{\mkbibemph{#1\isdot}}
  
  \AtEveryBibitem{%
    \ifentrytype{book}{}{
      \clearname{editor}
    }
  }
  \newbibmacro*{bbx:parunit}{%
    \ifbibliography
      {\setunit{\bibpagerefpunct}\newblock
       \usebibmacro{pageref}%
       \clearlist{pageref}%
       \setunit{\adddot\par\nobreak}}
      {}
  }
  \renewbibmacro*{doi+eprint+url}{%
    \usebibmacro{bbx:parunit}
    \iftoggle{bbx:doi}
      {\printfield{doi}}
      {}%
    \iftoggle{bbx:eprint}
      {\usebibmacro{eprint}}
      {}%
    \iftoggle{bbx:url}
      {\usebibmacro{url+urldate}}
      {}
  }
  \renewbibmacro*{eprint}{%
    \usebibmacro{bbx:parunit}
    \iffieldundef{eprinttype}
      {\printfield{eprint}}
      {\printfield[eprint:\strfield{eprinttype}]{eprint}}
  }
  \renewbibmacro*{url+urldate}{%
    \usebibmacro{bbx:parunit}
    \printfield{url}%
    \iffieldundef{urlyear}
      {}
      {\setunit*{\addspace}%
       \printtext[urldate]{\printurldate}}
  }
}{}\makeatother

\declaretheorem[numberwithin=section,refname={theorem,theorems},Refname={Theorem,Theorems}]{theorem}
\declaretheorem[sibling=theorem,style=definition]{definition}
\declaretheorem[sibling=theorem,name=Lemma]{lemma}
\declaretheorem[sibling=theorem,name=Proposition]{proposition}

\declaretheorem[sibling=theorem,style=definition,name=Remark]{remark}
\declaretheorem[sibling=theorem,style=definition,name=Example]{example}
\declaretheorem[sibling=theorem,name=Conjecture]{conjecture}
\declaretheorem[numbered=no,name=Question]{question}

\makeatletter\@ifpackageloaded{hyperref}{%
  \usepackage{xcolor}
  \definecolor{dark-red}{rgb}{0.4,0.15,0.15}
  \definecolor{dark-blue}{rgb}{0.15,0.15,0.4}
  \definecolor{medium-blue}{rgb}{0,0,0.5}
  \hypersetup{
    colorlinks,
    linkcolor={dark-red},
    citecolor={dark-blue},
    urlcolor={medium-blue}%
  }

}{}\makeatother

\usepackage{sectsty}
\sectionfont{\large\bfseries}
\subsectionfont{\normalsize\bfseries}

\providecommand{\abs}[1]{\lvert#1\rvert}

\providecommand{\floor}[1]{\lfloor#1\rfloor}
\providecommand{\Floor}[1]{\left\lfloor#1\right\rfloor}
\newcommand{\slope}{\pi}

\newcommand{\infw}[1]{%
  \ifcat\noexpand#1\relax\bm{#1}
  \else\mathbf{#1}\fi}          
\newcommand{\Lang}[2][]{\mathcal{L}_{#1}(#2)}
\newcommand{\pat}[2][]{\mathcal{P}_{#1}(#2)}
\DeclareMathOperator{\psw}{psw}
\newcommand{\oa}{\mathfrak{a}}
\newcommand{\ob}{\mathfrak{b}}
\newcommand{\oc}{\mathfrak{c}}
\newcommand{\N}{\mathbb{N}}
\newcommand{\Z}{\mathbb{Z}}

\newcommand{\keywords}[1]{\par\noindent{\footnotesize{\em Keywords\/}: #1}}

\begin{document}
  \title{Standard words and solutions of the word equation $X_1^2 \dotsm X_n^2 = (X_1 \dotsm X_n)^2$}
  \author[,1,2,3]{Jarkko Peltomäki\footnote{Corresponding author.\\E-mail addresses: \href{mailto:r@turambar.org}{r@turambar.org} (J. Peltomäki), \href{mailto:amsaar@utu.fi}{amsaar@utu.fi} (A. Saarela).}}
  \affil[1]{The Turku Collegium for Science and Medicine TCSM, University of Turku, Turku, Finland}
  \affil[2]{Turku Centre for Computer Science TUCS, Turku, Finland}
  \affil[3]{University of Turku, Department of Mathematics and Statistics, Turku, Finland}
  \author[3]{Aleksi Saarela}
  \date{}
  \maketitle
  \vspace{-1.5em}
  \noindent
  \hrulefill
  \begin{abstract}
    \vspace{-1em}
    \noindent
    We consider solutions of the word equation $X_1^2 \dotsm X_n^2 = (X_1 \dotsm X_n)^2$ such that the squares $X_i^2$
    are minimal squares found in optimal squareful infinite words. We apply a method developed by the second author for
    studying word equations and prove that there are exactly two families of solutions: reversed standard words and
    words obtained from reversed standard words by a simple substitution scheme. A particular and remarkable
    consequence is that a word $w$ is a standard word if and only if its reversal is a solution to the word equation
    and $\gcd(\abs{w}, \abs{w}_1) = 1$. This result can be interpreted as a yet another characterization for standard
    Sturmian words.

    We apply our results to the symbolic square root map $\sqrt{\cdot}$ studied by the first author and M. A.
    Whiteland. We prove that if the language of a minimal subshift $\Omega$ contains infinitely many solutions to the
    word equation, then either $\Omega$ is Sturmian and $\sqrt{\cdot}$-invariant or $\Omega$ is a so-called SL-subshift
    and not $\sqrt{\cdot}$-invariant. This result is progress towards proving the conjecture that a minimal and
    $\sqrt{\cdot}$-invariant subshift is necessarily Sturmian.

    \vspace{1em}
    \keywords{word equation, symbolic square root map, standard word, Sturmian word, optimal squareful word}
    \vspace{-1em}
  \end{abstract}
  \hrulefill

  \section{Introduction}
  Recently the second author of this paper solved a long-standing open problem by proving that if the equality of words
  $X^k = X_1^k \dotsm X_n^k$ holds for three positive values of $k$, then the words $X_1$, $\ldots$, $X_n$ commute
  \cite{2017:word_equations_where_a_power_equals_a_product,2019:word_equations_with_kth_powers_ofvariables}. If the
  equation is satisfied for at most two values of $k$, then noncommuting, or \emph{nonperiodic}, solutions can exist.
  In relation to Sturmian words, it was shown by the first author and M. A. Whiteland in
  \cite{2017:a_square_root_map_on_sturmian_words} (see also \cite{diss:jarkko_peltomaki}) that reversed standard words
  form a large nonperiodic solution class when $k = 1, 2$. More precisely, the research of
  \cite{2017:a_square_root_map_on_sturmian_words} concerns solutions of the word equation
  \begin{equation}\label{eq:square}
    X_1^2 \dotsm X_n^2 = (X_1 \dotsm X_n)^2
  \end{equation}
  such that the words $X_i$ are among the following six words for some fixed integers $\oa \geq 1$ and $\ob \geq 0$:
  \begin{alignat}{2}\label{eq:min_squares}
    &S_1 = 0,                 && S_4 = 10^\oa, \nonumber \\
    &S_2 = 010^{\oa-1}, \quad && S_5 = 10^{\oa+1}(10^\oa)^\ob, \\
    &S_3 = 010^\oa,           && S_6 = 10^{\oa+1}(10^\oa)^{\ob+1}. \nonumber
  \end{alignat}
  For example, the word $01010010$ is a solution when $\oa = 1$ and $\ob = 0$ because
  $(S_2 S_1 S_6)^2 = (01010010)^2 = (01)^2 \cdot 0^2 \cdot (10010)^2 = S_2^2 S_1^2 S_6^2$.

  Let us define standard words. Let $(d_k)$ be a sequence of positive integers, and define a sequence $(s_k)$ of words
  as follows:
  \begin{equation*}
    s_{-1} = 1, \quad s_0 = 0, \quad s_1 = s_0^{d_1-1} s_{-1}, \quad s_k = s_{k-1}^{d_k} s_{k-2} \quad \text{for $k \geq 2$}.
  \end{equation*}
  The words $s_k$ obtained in this manner are called \emph{standard words}. If $(d_k) = (2, 1, 1, 1, \ldots)$, then the
  words of the sequence $(s_k)$ are called \emph{Fibonacci words}. Notice that the word $01010010$ above is a reversed
  Fibonacci word (i.e., a Fibonacci word read from right to left). Notice also that the words \eqref{eq:min_squares} are
  reversed standard words (see \autoref{sec:preliminaries} for additional details).
  
  One of the main results of
  \cite{2017:a_square_root_map_on_sturmian_words} is that reversed standard words are solutions to \eqref{eq:square}.
  Another related solution class was also determined in \cite{2017:a_square_root_map_on_sturmian_words}: words obtained
  from reversed standard words by a certain substitution scheme. For example, consider the word $LSS$ and substitute $S$
  by a reversed standard word $w$ and $L$ by the word obtained from $w$ by exchanging its first two letters. If
  $w = 01010010$, then the resulting word $10010010 01010010 01010010$ is no longer a reversed standard word, but it is
  nevertheless a solution to \eqref{eq:square}. The main result of this paper is that there is no third solution type,
  that is, if we insist that the words $X_i$ in \eqref{eq:square} are among the words \eqref{eq:min_squares}, then a
  solution to \eqref{eq:square} is either
  \begin{itemize}
    \item a reversed standard word (a solution of type I) or
    \item obtained from a reversed standard word by a substitution scheme (a solution of type II).
  \end{itemize}
  For the precise statement, see \autoref{thm:main}. The essential component in our proofs is the method of assigning
  numerical values to letters and studying sums of letters geometrically developed by the second author in
  \cite{2017:word_equations_where_a_power_equals_a_product,2019:word_equations_with_kth_powers_ofvariables}. See also
  \cite{2018:studying_word_equations_by_a_method_of_weighted} and especially
  \cite{2018:an_optimal_bound_on_the_solution_sets_of_one-variable,2018:one-variable_word_equations_and_three-variable_constant-free}
  where the method was used to solve long-standing open problems on word equations.

  Our main result \autoref{thm:main} has the following surprising corollary.

  \begin{theorem}\label{thm:corollary}
    A word is standard if and only if its reversal is a type I solution to \eqref{eq:square}. In other words, a word
    $w$ is standard if and only if its reversal is a solution to \eqref{eq:square} and $\gcd(\abs{w}, \abs{w}_1) = 1$.
  \end{theorem}

  This is indeed remarkable given how strikingly different the definition of standard words is compared to
  \eqref{eq:square}.

  Standard words are the building blocks of the important and widely studied Sturmian words. Sturmian words are often
  defined as the infinite words having $n+1$ factors of each length $n$, but a more useful definition is that a
  Sturmian word is an infinite word that shares a language with a standard Sturmian word, and a standard Sturmian word
  is simply a limit of a sequence of standard words \cite[Proposition~2.2.24]{2002:algebraic_combinatorics_on_words}.
  Hence \autoref{thm:corollary} can be reinterpreted as follows.

  \begin{theorem}
    An infinite word is a standard Sturmian word if and only if it is a limit of the reversals of solutions $w$ to
    \eqref{eq:square} such that $\gcd(\abs{w}, \abs{w}_1) = 1$.\footnote{A tidier statement would be obtained if we reversed the words \eqref{eq:min_squares}, but then the interpretation of \eqref{eq:min_squares} as minimal squares appearing in optimal squareful words is lost; see one paragraph below.}
  \end{theorem}

  This is a surprising connection between a priori unrelated objects: standard Sturmian words can be characterized as
  the infinite words having $n+1$ factors of each length $n$ that are left-special
  \cite[Proposition~2.1.22]{2002:algebraic_combinatorics_on_words}. An infinite word $\infw{w}$ is \emph{left-special}
  if $a\infw{w}$ has the same language as $\infw{w}$ for all letters $a$ in the alphabet of $\infw{w}$.

  The specific solutions to \eqref{eq:square} considered in this paper were originally considered as a tool to
  construct fixed points for the symbolic square root map acting on optimal squareful words. An \emph{optimal squareful
  word} $\infw{w}$ is an aperiodic word such that each position of $\infw{w}$ begins with a square and the number of
  minimal squares occurring in $\infw{w}$ is the least possible. K. Saari proves in
  \cite{2010:everywhere_alpha-repetitive_sequences_and_sturmian_words} that an optimal squareful word $\infw{w}$
  contains exactly six minimal squares and there exists $\oa \geq 1$ and $\ob \geq 0$ such that each minimal
  square $X^2$ in $\infw{w}$ is such that $X$ is among the words \eqref{eq:min_squares}. Let $\infw{w}$ be an optimal
  squareful word and write it as a product of minimal squares: $\infw{w} = X_1^2 X_2^2 \dotsm$. The \emph{square root}
  $\sqrt{\infw{w}}$ of $\infw{w}$ is the word $X_1 X_2 \dotsm$ obtained by removing half of each square $X_i^2$.

  Saari showed that Sturmian words are optimal squareful
  \cite[Thm.~20]{2010:everywhere_alpha-repetitive_sequences_and_sturmian_words}, and the first author and Whiteland
  showed in \cite[Thm.~9]{2017:a_square_root_map_on_sturmian_words} that the square root map preserves the language of
  Sturmian words: if $\infw{s}$ is a Sturmian word, then typically $\infw{s} \neq \sqrt{\infw{s}}$, but $\infw{s}$ and
  $\sqrt{\infw{s}}$ have the same set of factors. This raised the question if other, non-Sturmian, and optimal
  squareful words with such a peculiar property exist. Clearly words having arbitrarily long prefixes that are squares
  of solutions to \eqref{eq:square} are fixed points of the square root map. The type II solutions of \eqref{eq:square}
  give in this way rise to non-Sturmian fixed points. If $\infw{u}$ is such a fixed point and $\infw{v}$ has the same
  language as $\infw{u}$, then typically $\sqrt{\infw{v}}$ has the same language as $\infw{v}$. Therefore the solutions
  of \eqref{eq:square} can be used to construct infinite words having interesting dynamics with respect to the square
  root map. The dynamics of these non-Sturmian words is further studied in
  \cite{2020:more_on_the_dynamics_of_the_symbolic_square_root}.

  While non-Sturmian words whose language is preserved by the square root map exist, Sturmian words satisfy a strong
  property: a Sturmian subshift is $\sqrt{\cdot}$-invariant (a subshift is Sturmian if it consists of Sturmian words).
  This property is not satisfied by the minimal subshifts related to the infinite words constructed from type II
  solutions (see \autoref{prp:sl_periodic}), and no further examples are known. We formulate a question of
  \cite{2017:a_square_root_map_on_sturmian_words} as the following conjecture stating a characterization of Sturmian
  subshifts.

   \begin{conjecture}\label{conjecture}
    Let $\Omega$ be a minimal optimal squareful subshift. Then $\Omega$ is a Sturmian subshift if and only if
    $\sqrt{\Omega} \subseteq \Omega$.
  \end{conjecture}

  We apply our main result \autoref{thm:main} and make progress towards this conjecture by proving the following
  result.
  
  \begin{theorem}
   Let $\Omega$ be a minimal subshift whose language contains infinitely many solutions to \eqref{eq:square} such that
   $\sqrt{\Omega} \subseteq \Omega$. Then $\Omega$ is Sturmian.
  \end{theorem}

  See \autoref{thm:main_square_root} for a slightly more general result. We leave \autoref{conjecture} open in the case
  that the language of $\Omega$ contains finitely many solutions to \eqref{eq:square}.

  The structure of the paper is as follows. The next section recalls preliminary notions and needed results. After
  this, we prove the main result \autoref{thm:main} in \autoref{sec:main}. In the following \autoref{sec:enumeration},
  we provide a formula for counting solutions to \eqref{eq:square} of length $n$. In \autoref{sec:applications},
  we apply the main results to the study of the square root map and prove \autoref{thm:main_square_root}. We end the
  paper by \autoref{sec:remarks} which contains additional results concerning the square root map.

  \section{Preliminaries}\label{sec:preliminaries}
  We use standard definitions and notation in combinatorics on words. The book
  \cite{2002:algebraic_combinatorics_on_words} is a standard reference for these concepts, and its second chapter is a
  standard reference for Sturmian words. Let $A$ be an \emph{alphabet}, i.e., a finite set of \emph{letters}, or
  \emph{symbols}. By concatenating the letters of $A$, we obtain the set of \emph{words} over $A$ denoted by $A^*$. The
  set $A^*$ contains the empty word $\varepsilon$, and we set $A^+ = A^* \setminus \{\varepsilon\}$. The length
  $\abs{w}$ of a word $w$ is the number of letters in $w$, and by $\abs{w}_a$ we mean the number of occurrences of the
  letter $a$ in $w$. A word $w$ is \emph{primitive} if $w = u^n$ only when $n = 1$. We often use the synchronization
  property of primitive word which states that a primitive word $w$ occurs in $w^2$ exactly twice: as a prefix and as a
  suffix. By a \emph{language} we simply mean a set of words, and by a language of a word we mean its set of factors.
  If $w$ is a word such that $\abs{w} \geq 2$, then by $L(w)$ we mean the word obtained from $w$ by exchanging its
  first two letters. A word $u$ is \emph{conjugate} to $v$ if there exists words $x$ and $y$ such that $u = xy$ and
  $v = yx$.
  
  We also consider infinite words over $A$ which are simply mappings $\N \to A$. An infinite word is \emph{purely
  periodic} if it is of the form $v^\omega$ and \emph{ultimately periodic} if it is of the form $uv^\omega$. An
  infinite word that is not ultimately periodic is called \emph{aperiodic}. If $\infw{w} = a_0 a_1 a_2 \dotsm$,
  $a_i \in A$, is an infinite word, then the \emph{shift} $T\infw{w}$ of $\infw{w}$ is the infinite word
  $a_1 a_2 \dotsm$. Let $\mathcal{L}$ be an extendable and factor-closed language. The set $\Omega$ of infinite words
  having language $\mathcal{L}$ is a \emph{subshift} with language $\mathcal{L}$. If $\mathcal{L}$ is the language of
  an infinite word $\infw{w}$, then we say that $\Omega$ is the subshift generated by $\infw{w}$, and we write
  $\Omega = \sigma(\infw{w})$. If every word in a subshift is aperiodic, then we call the subshift \emph{aperiodic}. A
  subshift is \emph{minimal} if it does not contain nonempty subshifts as proper subsets. A subshift is \emph{Sturmian}
  if all words in it are Sturmian words.

  We repeat the definition of optimal squareful words from the introduction. A square $w^2$ is \emph{minimal} if for
  each square prefix $u^2$ of $w^2$ it follows that $u = w$. Next we give the definition of squareful words; see
  \cite{2010:everywhere_alpha-repetitive_sequences_and_sturmian_words} for the more general definition of an
  everywhere $\alpha$-repetitive word.

  \begin{definition}
    An infinite word $\infw{w}$ is \emph{squareful} if each position of $\infw{w}$ begins with a square and the number
    of minimal squares occurring in $\infw{w}$ is finite. An infinite word $\infw{w}$ is \emph{optimal squareful} if it
    is aperiodic, squareful, and the number of distinct minimal squares in $\infw{w}$ is the least possible among
    aperiodic and squareful words.
  \end{definition}

  Saari proves in \cite[Thm.~16]{2010:everywhere_alpha-repetitive_sequences_and_sturmian_words} that if $\infw{w}$ is
  optimal squareful, then the minimal squares occurring in $\infw{w}$ are (up to renaming of letters) the squares of
  the words \eqref{eq:min_squares} for some fixed $\oa \geq 1$ and $\ob \geq 0$. In particular, a squareful word
  containing at most five distinct minimal squares is necessarily ultimately periodic. Saari characterizes optimal
  squareful words in \cite[Thm.~17]{2010:everywhere_alpha-repetitive_sequences_and_sturmian_words} as follows.

  \begin{proposition}\label{prp:optimal_squareful_characterization}
    An aperiodic infinite word $\infw{w}$ is optimal squareful if and only if (up to renaming of letters) there exist
    integers $\oa \geq 1$ and $\ob \geq 0$ such that $\infw{w}$ is an element of the language\footnote{The language $(u + v)^\omega$ consists of the infinitely long concatenations of the words $u$ and $v$.}
    \begin{equation*}
      0^* (10^\oa)^* (10^{\oa+1}(10^\oa)^\ob + 10^{\oa+1}(10^\oa)^{\ob+1})^\omega = S_1^* S_4^* (S_5 + S_6)^\omega.
    \end{equation*}
  \end{proposition}

  Here the symbols $S_i$ refer to the words \eqref{eq:min_squares}, and we assume this throughout the paper.
  Similarly we always use $\oa$ and $\ob$ to refer to the parameters of \eqref{eq:min_squares}; the parameters are
  assumed to be fixed. Moreover, we often write ``minimal square'' to mean a square of one of the words
  \eqref{eq:min_squares}. We refer to the words \eqref{eq:min_squares} themselves as \emph{minimal square roots}.

  By \autoref{prp:optimal_squareful_characterization}, optimal squareful words may initially contain arbitrarily high
  powers of $S_1$ and $S_4$. In the setting of the papers
  \cite{2017:a_square_root_map_on_sturmian_words,2020:more_on_the_dynamics_of_the_symbolic_square_root} this does not
  happen, and we continue this tradition.

  \begin{definition}
    The language $\Lang{\oa,\ob}$ consists of all factors of the infinite words in the language
    \begin{equation*}
      (10^{\oa+1}(10^\oa)^\ob + 10^{\oa+1}(10^\oa)^{\ob+1})^\omega = (S_5 + S_6)^\omega.
    \end{equation*}
  \end{definition}

  \begin{definition}
    The language $\Pi(\oa,\ob)$ consists of all nonempty words in $\Lang{\oa,\ob}$ that are products of the squares of
    the words \eqref{eq:min_squares}.
  \end{definition}

  Recall the word equation \eqref{eq:square}:
  \begin{equation*}
    X_1^2 X_2^2 \dotsm X_n^2 = (X_1 X_2 \dotsm X_n)^2.
  \end{equation*}
  We define specific solutions to this equation with respect to the language $\Pi(\oa, \ob)$.

  \begin{definition}
    A nonempty word $w$ is a \emph{solution to \eqref{eq:square}} if $w^2 \in \Pi(\oa, \ob)$ and $w$ can be written as
    a product of minimal square roots $w = X_1 X_2 \dotsm X_n$ which satisfy the word equation \eqref{eq:square}. The
    solution $w$ is \emph{primitive} if $w$ is a primitive word.
  \end{definition}

  Notice that the factorization of a square of a solution as a product of minimal squares is unique because none of the
  minimal squares is a prefix of another minimal square.

  \begin{example}
    Let $w$ be the word $01010010$, $\oa = 1$, and $\ob = 0$. Then $w^2 = 01010010 01010010 = 010 S_5 S_6^2$, so $w^2$
    is a suffix of $S_6 S_5 S_6^2$ meaning that $w^2 \in \Lang{\oa,\ob}$. Moreover, we have $w^2 = S_2^2 S_1^2 S_6^2$,
    so $w^2 \in \Pi(\oa, \ob)$. Since $w = S_2 S_1 S_6$, we see that $S_1^2 S_1^2 S_6^2 = (S_2 S_1 S_6)^2$. Thus $w$ is
    a solution to \eqref{eq:square}. In fact, the word $w$ is a primitive solution to \eqref{eq:square} since $w$ is a
    primitive word.
  \end{example}

  Let us then define the symbolic square root map defined in the introduction.

  \begin{definition}
    Factorize a word $w$ in $\Pi(\oa, \ob)$ as a product of minimal squares: $w = X_1^2 \dotsm X_n^2$. The \emph{square
    root} $\sqrt{w}$ of $w$ is defined as the word $X_1 \dotsm X_n$.
  \end{definition}

  In other words, a word $w$ is a solution to \eqref{eq:square} if and only if $\sqrt{w^2} = w$.

  Next we introduce the substitution scheme mentioned in the introduction as a means to build new solutions out of
  known solutions.

  \begin{definition}
    Let $u = a_0 \dotsm a_{n-1}$ with $a_i \in \{S, L\}$ be a nonempty word over the alphabet $\{S, L\}$.  We say that
    the word $u$ is a \emph{pattern word} if $a_i = a_j$ whenever $i$ and $j$ are in the same orbit\footnote{The
    numbers $i$ and $j$ are in the same orbit if there exists $k_1$ and $k_2$ such that
    $2^{k_1}i \equiv 2^{k_2}j \pmod{n}$.} of the mapping $x \mapsto 2x \bmod n$. We say that $u$ is \emph{nontrivial}
    if $\abs{u} > 1$.

    If $w$ is a nonempty binary word and $u$ is a word over $\{S, L\}$, then by $\pat[w]{u}$ we mean the word obtained
    from $u$ by replacing $S$ by $w$ and $L$ by $L(w)$.
  \end{definition}

  \begin{example}
    If $u = LSS$ or $u = SLLSLSS$, then $u$ is a pattern word. Whenever $w$ is a suitable solution to
    \eqref{eq:square}, then $\pat[w]{u}$ is also a solution to \eqref{eq:square} for a pattern word $u$ (see
    \autoref{prp:pattern_solution} below). For example, when $u = LSS$ and $w = 01010010$, then
    $\pat[w]{u} = 100100100101001001010010$ is a solution.
  \end{example}

  In what follows, we often use the symbol $S$ and the word $w$ interchangeably, but it will always be clear if $S$
  stands for the letter $S$ or a binary word. Notice that the map $\mathcal{P}_S$ is injective when the first two
  letters of $S$ are distinct.

  The definition of standard words was already given in the introduction. By a \emph{reversed} standard word, we mean a
  standard word read from right to left. Standard words are primitive; see
  \cite[Lemma~2.2.3]{2002:algebraic_combinatorics_on_words}. Consider then an integer sequence $(d_k)$ and the
  corresponding sequence $(s_k)$ of standard words. Replacing $(d_k)$ by $(d_2 + 1, d_3, d_4, \ldots)$ if necessary, we
  may assume that $d_1 \geq 2$ so that $11$ does not occur in the words $s_k$. It is not difficult to see that $d_1$
  (resp. $d_1 - 1$) is the maximum (resp. minimum) number of occurrences of the letter $0$ between two letters $1$ in
  the corresponding standard words $s_k$. Similarly $d_2$ (resp. $d_2 - 1$) indicates the maximum (resp. minimum)
  number of occurrences of $10^{d_1-1}$ between two occurrences of $10^{d_1}$. Hence we see that the words $s_k$ belong
  to the language $\Lang{\oa,\ob}$ with $\oa = d_1 - 1$ and $\ob = d_2 - 1$. The limit $\infw{s}$ of $(s_k)$ is a
  standard Sturmian word, and it is an optimal squareful word with parameters $\oa$ and $\ob$. Whenever we mention the
  words $S_i$ of \eqref{eq:min_squares} in relation to a standard word $w$, we understand that the parameters $\oa$ and
  $\ob$ related to $S_i$ and $w$ are the same. Given an integer sequence $(\oa+1, \ob+1, \ldots)$, the corresponding
  sequence of standard words begins as follows:
  \begin{align*}
    s_{-1} &= 1, \\
    s_0    &= 0, \\
    s_1    &= 0^{\oa}1, \\
    s_2    &= (0^{\oa}1)^{\ob+1} 0.
  \end{align*}
  Notice that the reversals of $s_0$ and $s_1$ equal the words $S_1$ and $S_4$ of \eqref{eq:min_squares}. Moreover, we
  have $S_2 = L(S_4)$, but $S_2$ is also a reversed standard word corresponding to $(\oa, 1, \ldots)$. The word $S_3$
  is a reversed standard word corresponding to $(\oa+1, 1, \ldots)$, and the word $L(S_3)$ is a standard word
  corresponding to $(\oa+2, \ldots)$. Similarly reversals of the words $S_5$ and $S_6$ correspond to standard words in
  the sequences $(\oa+1, \ob, 1, \ldots)$ and $(\oa+1, \ob+1, 1, \ldots)$. Moreover, the words $L(S_5)$ and $L(S_6)$
  are standard words. Therefore the words \eqref{eq:min_squares} found in a standard Sturmian word with parameters
  $\oa$ and $\ob$ are all reversals of standard words and some of them are related by the mapping $L$. It is
  straightforward to see that all reversed standard word of length at least $2$ begin with two distinct letters.

  \section{Characterization of Solutions}\label{sec:main}
  We can now formulate the following theorem which is the main result of this paper.

  \begin{theorem}\label{thm:main}
    Let $w$ be a primitive solution to \eqref{eq:square}. Then
    \begin{enumerate}[(I)]
      \item $w$ is a reversed standard word or
      \item there exists a reversed standard word $S$ with $\abs{S} > \abs{S_6}$ and a nontrivial and primitive pattern
            word $u$ such that $w = \pat[S]{u}$.
    \end{enumerate}
    Conversely, if (I) or (II) holds for a word $w$, then $w$ is a primitive solution to \eqref{eq:square}. Moreover,
    if $w$ is a nonprimitive solution to \eqref{eq:square}, then $w$ is a power of a primitive solution to
    \eqref{eq:square}.
  \end{theorem}

  We respectively call the two solution types of \autoref{thm:main} solutions of type I and type II. Observe also that
  a primitive pattern word always has odd length.
  
  \autoref{thm:main} implies the remarkable characterization of standard words stated in \autoref{thm:corollary}.

  \begin{proof}[Proof of \autoref{thm:corollary}]
    If $w$ is a standard word, then its reversal is a solution to \eqref{eq:square} by \autoref{thm:main}. It is a
    well-known property of standard words (see, e.g., the proof of
    \cite[Lemma~2.2.3]{2002:algebraic_combinatorics_on_words}) that $\gcd(\abs{w}, \abs{w}_1) = 1$. Suppose on the
    other hand that $w$ is a solution to \eqref{eq:square} and $\gcd(\abs{w}, \abs{w}_1) = 1$. Then $w$ must be
    primitive for otherwise $\gcd(\abs{w}, \abs{w}_1) > 1$. If $w$ is a solution of type II, then \autoref{thm:main}
    implies that there exists a reversed standard word $S$ and  a nontrivial and primitive pattern word $u$ such that
    $w = \pat[S]{u}$. Since $\abs{S}_1 = \abs{L(S)}_1$, we see that $\abs{w} = \abs{u}\abs{S}$ and
    $\abs{w}_1 = \abs{u}\abs{S}_1$, so $\gcd(\abs{w}, \abs{w}_1) \geq \abs{u}$. Thus it must be that $\abs{u} = 1$, but
    this contradicts the fact that $u$ is nontrivial. Hence $w$ cannot be of type II, so it is of type I, that is, $w$
    is a reversed standard word.
  \end{proof}

  Notice that the preceding proof also shows that the sets of type I solutions and type II solutions are disjoint.

  Before showing that solutions to \eqref{eq:square} are of the claimed form, we present results showing that words
  satisfying (I) or (II) of \autoref{thm:main} are indeed solutions. The case (I) is handled by the following result.

  \begin{proposition}\label{prp:standard_solutions}\cite[Proposition~23]{2017:a_square_root_map_on_sturmian_words}
    If $w$ is a reversed standard word, then $w$ is a solution to \eqref{eq:square}.
  \end{proposition}

  For the case (II) (see \autoref{prp:pattern_solution}), we need the following lemmas.

  \begin{lemma}\label{lem:sl_square_root}\cite[Lemma~22]{2017:a_square_root_map_on_sturmian_words}
    Let $S$ be a reversed standard word such that $\abs{S} > \abs{S_6}$, and set $L = L(S)$. Then
    $SS, SL, LS, LL \in \Pi(\oa, \ob)$, $\sqrt{SS} = \sqrt{SL} = S$, and $\sqrt{LL} = \sqrt{LS} = L$.
  \end{lemma}

  \begin{lemma}\label{lem:pat_primitive}
    Let $u$ be a primitive word over $\{S, L\}$ and $w$ a reversed standard word such that $\abs{w} > 1$. Then
    $\pat[w]{u}$ is primitive.
  \end{lemma}
  \begin{proof}
    Let $\pat[w]{u} = v^k$ for a primitive word $v$ and integer $k \geq 1$. From $\abs{w} = \abs{L(w)}$ and
    $\abs{w}_1 = \abs{L(w)}_1$, it follows that $\abs{\pat[w]{u}} = \abs{u} \cdot \abs{w}$ and
    $\abs{\pat[w]{u}}_1 = \abs{u} \cdot \abs{w}_1$. Because $w$ is a reversed standard word, we have
    $\gcd(\abs{w}, \abs{w}_1) = 1$, and therefore
    \begin{equation*}
      \gcd(\abs{\pat[w]{u}}, \abs{\pat[w]{u}}_1)
      = \gcd(\abs{u} \cdot \abs{w}, \abs{u} \cdot \abs{w}_1)
      = \abs{u} \gcd(\abs{w}, \abs{w}_1) = \abs{u}.
    \end{equation*}
    On the other hand,
    \begin{equation*}
      \gcd(\abs{\pat[w]{u}}, \abs{\pat[w]{u}}_1)
      = \gcd(\abs{v^k}, \abs{v^k}_1)
      = \gcd(k \abs{v}, k \abs{v}_1)
      = k \gcd(\abs{v}, \abs{v}_1).
    \end{equation*}
    Thus $\abs{u}$ is a multiple of $k$, and we can write $u = u_1 \dotsm u_k$ with $\abs{u_1} = \ldots = \abs{u_k}$.
    Then
    \begin{equation*}
      v^k = \pat[w]{u} = \pat[w]{u_1} \dotsm \pat[w]{u_k}
      \qquad \text{and} \qquad
      \abs{\pat[w]{u_1}} = \ldots = \abs{\pat[w]{u_k}},
    \end{equation*}
    and therefore $v = \pat[w]{u_i}$ for all $i$. By the injectivity of $\mathcal{P}_w$, we see that
    $u_1 = \ldots = u_k$ and $u = u_1^k$. Because $u$ is primitive, it must be that $k = 1$. Since $\pat[w]{u} = v^k$,
    we conclude that $\pat[w]{u}$ is primitive.
  \end{proof}

  Parts of the following result and its proof appear in less general form in
  \cite[Lemma~39]{2017:a_square_root_map_on_sturmian_words}, \cite[Lemma~40]{2017:a_square_root_map_on_sturmian_words},
  and \cite[Lemma~2.8]{2020:more_on_the_dynamics_of_the_symbolic_square_root}.

  \begin{proposition}\label{prp:pattern_solution}
    Let $S$ be a reversed standard word such that $\abs{S} > \abs{S_6}$. Then $\pat[S]{u}$ is a solution to
    \eqref{eq:square} for any pattern word $u$. Moreover, if $u$ is primitive, then $\pat[S]{u}$ is primitive.
  \end{proposition}
  \begin{proof}
    Let $u$ be a pattern word. Write $u^2$ as blocks of two letters:
    $u^2 = A_0 B_0 \cdot A_1 B_1 \dotsm A_{\abs{u}-1} B_{\abs{u}-1}$. Then
    $\smash[t]{\sqrt{\pat[S]{A_i B_i}} = \pat[S]{A_i}}$ for all $i$ by \autoref{lem:sl_square_root}. Consequently
    \begin{equation}\label{eq:foo}
      \sqrt{\pat[S]{u^2}} = \sqrt{\pat[S]{A_0 B_0 \dotsm A_i B_i \dotsm A_{\abs{u}-1} B_{\abs{u}-1}}} = \pat[S]{A_0 \dotsm A_i \dotsm A_{\abs{u}-1}}.
    \end{equation}
    The letter $A_i$ is the $2i$th letter of $u^2$. Since $u$ is a pattern word, it follows that $A_i$ equals the $i$th
    letter of $u$. Thus \eqref{eq:foo} states that $\sqrt{\pat[S]{u^2}} = \pat[S]{u}$. This means that $\pat[S]{u}$ is a
    solution to \eqref{eq:square}.
    
    %

    If $u$ is primitive, then $\pat[S]{u}$ is primitive by \autoref{lem:pat_primitive}.
  \end{proof}

  Let us then begin proving the converse of \autoref{prp:standard_solutions} and \autoref{prp:pattern_solution}. The
  method used is to assign numerical values to letters as mentioned in the introduction.

  If we assign distinct real weights to the letters $0$ and $1$, say $\omega_0$ and $\omega_1$, then we may define the
  \emph{sum} $\Sigma(u)$ of $u = a_1 \dotsm a_n \in \{0, 1\}^*$ to be the real number
  \begin{equation}\label{eq:sigma}
      \Sigma(u) = \omega_{a_1} + \dotsm + \omega_{a_n}
      = \abs{u}_0 \cdot \omega_0 + \abs{u}_1 \cdot \omega_1.
  \end{equation}
  Different weights $\omega_0$ and $\omega_1$ obviously give different sum functions $\Sigma$. We want to choose the
  weights $\omega_0$ and $\omega_1$ so that $\Sigma(w) = 0$ for a certain fixed word $w$ containing both letters $0$
  and $1$. Moreover, we want to normalize the weights so that $\omega_1 - \omega_0 = 1$. Both of these conditions are
  satisfied if we choose
  \begin{equation}\label{eq:omega}
    \omega_0 = -\frac{\abs{w}_1}{\abs{w}}
    \quad \text{and} \quad
    \omega_1 = \frac{\abs{w}_0}{\abs{w}}.
  \end{equation}
  In what follows, we fix a word $w$ that is a solution to \eqref{eq:square}, and then assume that $\omega_0$,
  $\omega_1$, and $\Sigma$ have been defined as in \eqref{eq:omega} and \eqref{eq:sigma}.

  The \emph{slope} of $u \in \{0, 1\}^+$ is $\slope(u) = \abs{u}_1 / \abs{u}$. We can represent $\Sigma$ also with the
  help of the function $\slope$:
  \begin{equation}\label{eq:weight-slope}
    \Sigma(u) = \abs{u} (\slope(u) - \slope(w)).
  \end{equation}
  This is shown by the following computation:
  \begin{align*}
    \Sigma(u)
    &= -\abs{u}_0 \cdot \frac{\abs{w}_1}{\abs{w}} + \abs{u}_1 \cdot \frac{\abs{w}_0}{\abs{w}}
    = -\abs{u}_0 \slope(w) + \abs{u}_1 \cdot \frac{\abs{w} - \abs{w}_1}{\abs{w}}\\
    &= -\abs{u}_0 \slope(w) + \abs{u}_1 - \abs{u}_1 \slope(w)
    = \abs{u}_1 - (\abs{u}_0 + \abs{u}_1) \slope(w)
    = \abs{u} \slope(u) - \abs{u} \slope(w).
  \end{align*}

  Let $u = a_1 \dotsm a_n$ be a word over $\{0, 1\}$. We define the \emph{prefix sum word} $\psw(u)$ of $u$ to be the
  word $\psw(u) = b_1 \dotsm b_n$, where $b_i = \Sigma(a_1 \dotsm a_i)$ for all $i$. This is a word over some alphabet
  that is a subset of the rational numbers. Naturally, we can denote the largest and smallest letters in $\psw(u)$ by
  $\max(\psw(u))$ and $\min(\psw(u))$, respectively. The word $u$ has a graphical representation as a plane curve that
  we get by connecting the points $(0, 0)$, $(1, b_1)$, $\ldots$, $(n, b_n)$.

  \begin{example}
    Let $w = 01010010$. Then $\omega_0 = -3 / 8$, $\omega_1 = 5 / 8$, and
    \begin{equation*}
        \psw(w) =
        \frac{-3}{8}, \frac{2}{8}, \frac{-1}{8}, \frac{4}{8},
        \frac{1}{8}, \frac{-2}{8}, \frac{3}{8}, \frac{0}{8},
    \end{equation*}
    where we have used commas between the letters and the same denominator $8$ in every letter for clarity. See
    \autoref{fig:curve} for a graphical representation.
  \end{example}

  \begin{figure}[th]
    \centering
    \begin{tikzpicture}[scale=1.2]
        \draw[thick]
            (0, 0) node{$\bullet$}
            -- ++(1, -0.375) node{$\bullet$}
            -- ++(1, 0.625) node{$\bullet$}
            -- ++(1, -0.375) node{$\bullet$}
            -- ++(1, 0.625) node{$\bullet$}
            -- ++(1, -0.375) node{$\bullet$}
            -- ++(1, -0.375) node{$\bullet$}
            -- ++(1, 0.625) node{$\bullet$}
            -- ++(1, -0.375) node{$\bullet$}
            ;
        \draw[->] (-1, 0) -- (9, 0);
        \draw[->] (0, -1) -- (0, 1);
    \end{tikzpicture}
    \caption{The curve of the word $01010010$.}\label{fig:curve}
  \end{figure}
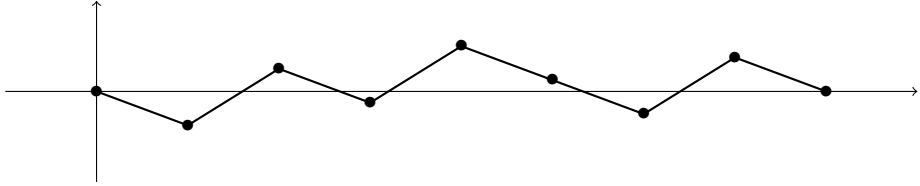
 
  \begin{lemma}\label{lem:six-slope}
    Let $x \in \{S_1, S_2, S_3, S_4, S_5, S_6\}$.
    \begin{enumerate}
      \item If $x = y10z$ and $y \ne \varepsilon$, then $\slope(y) < \slope(10z)$.
      \item If $x = y01z$ and $y \ne \varepsilon$, then $\slope(y) > \slope(01z)$.
    \end{enumerate}
  \end{lemma}
  \begin{proof}
    In the first case, either $y = 0$ and $10z = 10^i$ for some $i$ or $y = 10^{\oa+1} (10^\oa)^i$ and
    $10z = (10^\oa)^j$ for some $i, j$. The former option is clear: $\slope(y) = 0 < \slope(10z)$. In the latter case,
    we have $\slope(10z) = 1/(\oa+1) > (i + 1)/( (i+1)(\oa+1) + 1) = \slope(y)$.

    In the second case, the only possibility is $y \in \{10^\oa, 10^{\oa+1} (10^\oa)^i 10^{\oa-1}\}$ and
    $01z = 0 (1 0^\oa)^j$ for some $i, j$, and then $\slope(y) = 1 / (\oa + 1) = \slope(1z) > \slope(01z)$.
  \end{proof}

  The following lemma states that the curve of a solution $w$ is contained in the rather small space between the lines
  $y = \omega_0$ and $y = \omega_1$.

  \begin{lemma}\label{lem:narrow_tube}
    Let $w$ be a solution to \eqref{eq:square} containing both letters $0$ and $1$. Then
    $\omega_0 \leq \min(\psw(w)) \leq \max(\psw(w)) \leq \omega_1$.
  \end{lemma}
  \begin{proof}
    Write $w^2$ as a product of minimal squares: $w^2 = X_1^2 \dotsm X_n^2$. First we want to argue that if the maximum
    of $\psw(w)$ occurs at a position between $\abs{X_1 \dotsm X_r}$ and $\abs{X_1 \dotsm X_r X_{r+1}}$, then
    $\Sigma(X_1 \dotsm X_r) \leq 0$ and $\Sigma(X_1 \dotsm X_{r+1}) \leq 0$. Say
    $\max(\psw(w)) = \Sigma(X_1 \dotsm X_r u)$ and $X_{r+1} = uv$ (we allow here $r = 0$, and then
    $X_1 \dotsm X_r = \varepsilon$). Assume for a contradiction that $\Sigma(X_1 \dotsm X_r) > 0$. Since $w$ is
    zero-sum, we have $\max(\psw(w^2)) = \max(\psw(w))$. The word $X_1^2 \dotsm X_r^2 u$ is a prefix of $w^2$. Now
    \begin{equation*}
      \Sigma(X_1^2 \dotsm X_r^2 u) = 2\Sigma(X_1 \dotsm X_r) + \Sigma(u) > \Sigma(X_1 \dotsm X_r) + \Sigma( u) = \max(\psw(w)),
    \end{equation*}
    which contradicts the fact that $\max(\psw(w^2)) = \max(\psw(w))$. Hence $\Sigma(X_1 \dotsm X_r) \leq 0$. Suppose
    then that $\Sigma(X_1 \dotsm X_{r+1}) > 0$. The word $X_1^2 \dotsm X_r^2 X_{r+1} u$ is a prefix of $w^2$. We
    compute:
    \begin{equation*}
      \Sigma(X_1^2 \dotsm X_r^2 X_{r+1} u)
      = \Sigma(X_1 \dotsm X_{r+1}) + \Sigma(X_1 \dotsm X_r u)
      > \Sigma(X_1 \dotsm X_r u) = \max(\psw(w)).
    \end{equation*}
    This is impossible, so $\Sigma(X_1 \dotsm X_{r+1}) \leq 0$.

    A symmetric argument shows that if $\psw(w)$ attains its minimum value at $X_1 \dotsm X_r u$, then
    $\Sigma(X_1 \dotsm X_r) \geq 0$ and $\Sigma(X_1 \dotsm X_{r+1}) \geq 0$ when $r$ and $u$ are defined like above.

    Let us proceed to consider the case of the maximum value $\max(\psw(w))$. Let
    $\max(\psw(w)) = \Sigma(X_1 \dotsm X_r u)$ and $X_{r+1} = uv$. We have shown above that
    $\Sigma(X_1 \dotsm X_r) \leq 0$ and $\Sigma(X_1 \dotsm X_{r+1}) \leq 0$. If
    $\abs{u} \in \{0, 1, \abs{X_{r + 1}}\}$, then it is easy to see that $\max(\psw(w)) \leq \omega_1$. Otherwise, we
    can write $u = ya$ and $v = bz$, where $a, b \in \{0, 1\}$ and $y \ne \varepsilon$. If $a = 0$, then
    \begin{equation*}
      \Sigma(X_1 \dotsm X_r u) = \Sigma(X_1 \dotsm X_r y) + \omega_0 < \Sigma(X_1 \dotsm X_r y),
    \end{equation*}
    which contradicts $\max(\psw(w)) = \Sigma(X_1 \dotsm X_r u)$, so it must be $a = 1$. Similarly, if $b = 1$, then
    \begin{equation*}
        \Sigma(X_1 \dotsm X_r u1) = \Sigma(X_1 \dotsm X_r u) + \omega_1 > \Sigma(X_1 \dotsm X_r u)
    \end{equation*}
    contradicts $\max(\psw(w)) = \Sigma(X_1 \dotsm X_r u)$, so it must be $b = 0$. By \autoref{lem:six-slope}, we have
    $\slope(y) < \slope(10z)$. Therefore, both $\slope(w) < \slope(y)$ and $\slope(10z) < \slope(w)$ cannot hold. In
    other words, at least one of $\slope(y) \leq \slope(w)$ and $\slope(10z) \geq \slope(w)$ is true. In the former
    case, $\Sigma(y) \leq 0$ by \eqref{eq:weight-slope}, and thus
    \begin{equation*}
      \max(\psw(w)) = \Sigma(X_1 \dotsm X_r) + \Sigma(y1) \leq \omega_1.
    \end{equation*}
    In the latter case, $\Sigma(10z) \geq 0$ by \eqref{eq:weight-slope}, and so
    \begin{equation*}
      \max(\psw(w)) = \Sigma(X_1 \dotsm X_{r+1}) - \Sigma(0z) \leq \omega_1.
    \end{equation*}

    Consider then the minimum $\min(\psw(w))$. Using the above notation, we have shown that
    $\Sigma(X_1 \dotsm X_r) \geq 0$ and $\Sigma(X_1 \dotsm X_{r+1}) \geq 0$. If
    $\abs{u} \in \{0, 1, \abs{X_{r+1}}\}$, then it is again easy to see that $\min(\psw(w)) \geq \omega_0$. Otherwise,
    we can write $u = y0$ and $v = 1z$, where $y \ne \varepsilon$. Again \autoref{lem:six-slope} implies that
    $\slope(y) \geq \slope(w)$ or $\slope(01z) \leq \slope(w)$. In the former case, $\Sigma(y) \geq 0$ by
    \eqref{eq:weight-slope}, and hence
    \begin{equation*}
      \min(\psw(w)) = \Sigma(X_1 \dotsm X_r) + \Sigma(y0) \geq \omega_0.
    \end{equation*}
    In the latter case, $\Sigma(01z) \leq 0$ by \eqref{eq:weight-slope}, and thus
    \begin{equation*}
      \max(\psw(w)) = \Sigma(X_1 \dotsm X_{r+1}) - \Sigma(1z) \geq \omega_0.
    \end{equation*}
  \end{proof}

  \begin{lemma}\label{lem:SL_blocks}
    Let $w$ be a solution to \eqref{eq:square} containing both letters $0$ and $1$. Let $\slope(w) = c/d$, where $c$
    and $d$ are relatively prime positive integers. Then $w \in \{01u, 10u\}^+$, where $u = a_1 \dotsm a_{d-2}$ and
    \begin{equation}\label{eq:central}
      a_j = \Floor{\frac{c (j + 1)}{d}}
            - \Floor{\frac{c j}{d}}
    \end{equation}
    for all $j$.
  \end{lemma}
  \begin{proof}
    The point of the proof is that after the initial letter of $w$ is chosen, the remaining letters are uniquely
    determined by the number $\slope(w)$. Indeed, if $v$ is a prefix of $w$ and $\Sigma(v) > 0$, then $v0$ is a prefix
    of $w$ or otherwise \autoref{lem:narrow_tube} is violated. If $\Sigma(v) < 0$, then $v1$ is a prefix of $w$.

    Write $w = z_1 \dotsm z_m$ where for all $i$ we have $\Sigma(z_i) = 0$ and $\Sigma(p) \neq 0$ for all nonempty
    proper prefixes $p$ of $z_i$. We are going to show that every $z_i$ is in $\{01u, 10u\}$, which proves the theorem.
    By \autoref{lem:narrow_tube}, $\omega_0 \leq \Sigma(p) \leq \omega_1$ for all prefixes $p$ of $z_i$. By
    \eqref{eq:weight-slope}, $\Sigma(z_i) = 0$ is equivalent to $\slope(z_i) = \slope(w)$ which implies
    $\abs{z_i} \geq d \geq 2$. Because $\Sigma(00) < \omega_0$ and $\Sigma(11) > \omega_1$, $z_i$ must begin with
    either $01$ or $10$. Let the following letters after that be $b_1, \dots, b_{d - 2}$. Let
    $k_j = \abs{b_1 \dotsm b_j}_1$ for all $j$. Then
    \begin{align*}
           & \omega_0 \leq \Sigma(01 b_1 \dotsm b_j) \leq \omega_1 \\
      \iff & \omega_0 \leq (j - k_j + 1) \omega_0 + (k_j + 1) \omega_1 \leq \omega_1 \\
      \iff & -(j + 1)\omega_0 + \omega_0 - \omega_1 \leq k_j (\omega_1 - \omega_0) \leq -(j + 1)\omega_0 \\
      \iff & (j + 1) \dfrac{c}{d} - 1 \leq k_j \leq (j + 1) \dfrac{c}{d}.
    \end{align*}
    Here we have used the facts $\omega_1 - \omega_0 = 1$ and $\omega_0 = -\slope(w) = -c/d$. If $j \leq d - 2$, then
    $(j + 1) c / d$ is not an integer, and thus $k_j = \floor{(j + 1) c / d}$. It follows that
    $b_j = k_j - k_{j - 1} = a_j$ for all $j \in \{1, \dots, d - 2\}$. Consequently, we have shown that $z_i$ begins
    with either $01u$ or $10u$. Further, we have
    \begin{equation*}
      \slope(01u)
      = \frac{1 + k_{d - 2}}{d}
      = \frac{1 + \floor{(d - 1) c / d}}{d}
      = \frac{c}{d},
    \end{equation*}
    so $\Sigma(01 u) = \Sigma(10 u) = 0$ by \eqref{eq:weight-slope}. Because $z_i$ does not have nonempty proper
    prefixes with zero sum, it must be that $z_i \in \{01u, 10u\}$.
  \end{proof}

  The formula \eqref{eq:central} for the word $u$ of \autoref{lem:SL_blocks} matches exactly the construction of
  so-called central words. Usually it is defined that a binary word $w$ is \emph{central} if $w01$ and $w10$ are
  standard words. In \cite[Ch.~2.2.1]{2002:algebraic_combinatorics_on_words}, a central word of length $d$ containing
  $c$ occurrences of $1$ with $c$, $d$ relatively prime is constructed. This construction uses the same formula as
  \eqref{eq:central}, so \cite[Prop.~2.2.12]{2002:algebraic_combinatorics_on_words}, which proves the validity of the
  construction, shows that the word $u$ of \autoref{lem:SL_blocks} is a central word. Thus $u01$ and $u10$ are standard
  words. Since central words are palindromes (see \cite[Thm.~2.2.4]{2002:algebraic_combinatorics_on_words}), it follows
  that the words $01u$ and $10u$ are reversed standard words. We have thus proved the following result.

  \begin{proposition}\label{prp:sl_product}
    Let $w$ be a solution to \eqref{eq:square} containing both letters $0$ and $1$. Then there exists a unique reversed
    standard word $S$ such that $w \in \{S, L(S)\}^+$.
  \end{proposition}

  We need one small result before we can prove \autoref{thm:main}. When we use this result in the proof of
  \autoref{thm:main}, the word $w'$ will actually be $w^2$.

  \begin{lemma}\label{lem:square_prefix}
    Let $w'$ in $\Pi(\oa, \ob)$ be a word such that $\sqrt{w'}$ is a prefix of $w'$. If $u^2$ is a prefix of $w'$, then
    $u$ is a solution to \eqref{eq:square}.
  \end{lemma}
  \begin{proof}
    Write $w'$ as a product of minimal squares: $w' = X_1^2 \dotsm X_n^2$, and let $m$ be the largest index such that
    $X_1^2 \dotsm X_m^2$ is a prefix of $u^2$. Write $u^2 = X_1^2 \dotsm X_m^2 z$ for a prefix $z$ of $X_{m+1}^2$. It
    follows that $\abs{z}$ is even, so we may write $z = xy$ with $\abs{x} = \abs{y}$. Since $u$ and
    $X_1 \dotsm X_m X_{m+1}$ are prefixes of $w'$, $x$ is a prefix of $X_{m+1}$, and
    $\abs{u} = \abs{X_1 \dotsm X_m x}$, we see that $u = X_1 \dotsm X_m x$. Hence $x$ is a suffix of $u$ and $y = x$,
    that is, $z = x^2$. Since the square $X_{m+1}$ is minimal and the index $m$ is maximal, the only option is that $z$
    is empty. Thus $u^2 = X_1^2 \dotsm X_m^2$ and $u = X_1 \dotsm X_m$. In other words, the word $u$ is a solution to
    \eqref{eq:square}.
  \end{proof}

  \begin{proof}[Proof of \autoref{thm:main}]
    Let $w$ be a primitive solution to \eqref{eq:square}. If $\abs{w} = 1$, then $w = 0$ and $w$ is a reversed standard
    word. We may thus suppose that $\abs{w} > 1$. Since $w$ is primitive, this means that both letters $0$ and $1$
    occur in $w$. By \autoref{prp:sl_product}, there exists a reversed standard word $S$ such that $w \in \{S, L\}^+$
    where $L = L(S)$. The word $w$ can be understood as a word over the alphabet $\{S, L\}$; we denote this word by
    $u$. If $\abs{u} = 1$, then the case (I) applies, so assume that $\abs{u} > 1$. For the claim, it suffices to
    establish that $u$ is a primitive pattern word and $\abs{S} > \abs{S_6}$. Suppose first that $\abs{S} > \abs{S_6}$.
    Group the letters of $u^2$ as blocks of two: $u^2 = A_0 B_0 \cdot A_1 B_1 \dotsm A_{\abs{u}-1} B_{\abs{u}-1}$. If
    $AB$ is such a block, then $\smash[t]{\sqrt{\pat[S]{AB}} = \pat[S]{A}}$ by \autoref{lem:sl_square_root}. Then
    \begin{equation*}
      \sqrt{\pat[S]{A_0 B_0 \dotsm A_i B_i}} = \pat[S]{A_0 \dotsm A_i}
    \end{equation*}
    for all $i$. Since $w$ is a solution, the word $\pat[S]{A_0 \dotsm A_i}$ is a prefix of $w$. Since $\mathcal{P}_S$
    is injective, we conclude that the $i$th letter of $u$ equals its $2i$th letter when the indices are understood
    modulo $\abs{u}$. Therefore $u$ is a pattern word. The word $u$ must be primitive because otherwise $\pat[S]{u}$ is
    not primitive.

    The next part of the proof consists of showing that $w^2$ is not in $\Pi(\oa, \ob)$ if $\abs{S} \leq \abs{S_6}$
    meaning that $w$ is not a solution. Recall from \autoref{sec:preliminaries} the construction of reversed standard
    words having parameters $\oa$ and $\ob$ and the fact that the words \eqref{eq:min_squares} are reversed standard
    words. It is straightforward to see that if $S$ is a reversed standard word having parameters $\oa$ and $\ob$ and
    containing both letters $0$ and $1$ and $\abs{S} \leq \abs{S_6}$, then $S$ is of the form
    \begin{equation*}
      10^\oa, \quad 0(10^\oa)^\ell \quad \text{with $1 \leq \ell \leq \ob + 1$}, \quad 10^{\oa+1} (10^\oa)^{\ob+1}
    \end{equation*}
    up to an application of the mapping $L$.

    Since the word $u$ is primitive and $\abs{u} > 1$, both letters $S$ and $L$ occur in $u$. It follows that both $SL$
    and $LS$ occur in $u^2$. If $S = 10^\oa$, then $L = 010^{\oa-1}$ and $LS$ contains the factor $010^{\oa-1} 1$
    showing that $\pat[S]{LS} \notin \Lang{\oa, \ob}$. By symmetry, the same happens if $S = 010^{\oa-1}$. Similar
    reasoning in the remaining cases shows that $LS$ or $SL$ must contain a factor that contradicts with
    $w^2 \in \Lang{\oa, \ob}$. Consequently $w^2 \notin \Pi(\oa, \ob)$.

    If (I) or (II) holds for $w$, then \autoref{prp:standard_solutions} and \autoref{prp:pattern_solution} respectively 
    imply that $w$ is a primitive solution to \eqref{eq:square}. Let finally $w$ be a nonprimitive solution to
    \eqref{eq:square}, and write $w = v^k$ for a primitive word $v$ and integer $k \geq 2$. By
    \autoref{lem:square_prefix}, the word $v$ is a solution to \eqref{eq:square}. Thus $w$ is a power of a primitive
    solution to \eqref{eq:square}.
  \end{proof}

  \section{Enumeration of Solutions}\label{sec:enumeration}
  In this section, we give a formula for counting solutions to \eqref{eq:square} of length $n$ for all possible values
  of $\oa$ and $\ob$.

  As before, we assume that a solution never contains the factor $11$, that is, we count the solutions up to the
  isomorphism $0 \mapsto 1$, $1 \mapsto 0$. Since $L(w)$ is a solution whenever $w$ is a solution and $\abs{w} \geq 2$,
  we may count the solutions up to the application of the mapping $L$ as well. Let us first recall an important result.

  \begin{proposition}\label{prp:euler}\cite[Cor.~2.2.16]{2002:algebraic_combinatorics_on_words}
    The number of standard words of length $n$ up to the isomorphism $0 \mapsto 1$, $1 \mapsto 0$ is given by
    $\varphi(n)$, where $\varphi$ is Euler's totient function.
  \end{proposition}

  Before counting the number of solutions, we need to count the number of pattern words of length $\ell$ that begin
  with the letter $S$. From the definition of a pattern word, it is clear that this number is
  $2^{\mathcal{O}(\ell) - 1}$ where $\mathcal{O}(\ell)$ is the number of orbits of the mapping
  $x \mapsto 2x \bmod{\ell}$. The number $\mathcal{O}(\ell)$ depends heavily on the arithmetic nature of $\ell$, and we
  do not study it in more detail in this paper. Suffice it to say that the first values of $\mathcal{O}(\ell)$ are $1$,
  $1$, $2$, $1$, $2$, $2$, $3$, $1$, $3$, $2$, $2$, $2$, $2$, $3$, $5$, $1$, $3$, $3$, $2$, $2$, $6$ (see the entry
  \href{https://oeis.org/A000374}{A000374} in Sloane's \emph{On-Line Encyclopedia of Integer Sequences} \cite{oeis}),
  and the following lemma provides a formula for $\mathcal{O}(\ell)$.

  \begin{lemma}
    Let $\ell$ be a positive integer and $2^i$ the largest power of $2$ dividing $\ell$, and set $\ell' = \ell / 2^i$.
    We have
    \begin{equation*}
      \mathcal{O}(\ell) = \sum_{d \mid \ell'} \frac{\varphi(d)}{\operatorname{ord}(2, d)}
    \end{equation*}
    where $\operatorname{ord}(2, d)$ is the order of the element $2$ in the group $\Z^*_d$.
  \end{lemma}
  \begin{proof}
    Each orbit of the mapping $x \mapsto 2x \bmod{\ell}$ contains a unique cycle, so it suffices to compute the number
    of cycles. If $x$, $2x$, $\ldots$, $2^k x$ is a cycle (that is, $2^{k+1} x = x$), then
    $(2^{k+1} - 1)x \equiv 0 \pmod{\ell}$, and it must be that $2^i$ divides $x$. By dividing $x$ by $2^i$, we thus
    obtain a cycle of length $k + 1$ modulo $\ell / 2^i$. Conversely every cycle modulo $\ell / 2^i$ produces a cycle
    of the same length modulo $\ell$ so, in order to count the cycles, it suffices to consider the case $i = 0$.

    Let $d = \gcd(x, \ell)$. Then $x/d$, $2x/d$, $\ldots$, $2^k x/d$ is a cycle of length $k+1$ modulo $\ell/d$, and it
    corresponds to a coset of the subgroup $\langle 2 \rangle$ generated by $2$. Thus
    $k + 1 = \operatorname{ord}(2, \ell/d)$. Conversely every coset of $\langle 2 \rangle$ yields a cycle of length
    $\operatorname{ord}(2, \ell/d)$ modulo $\ell$. The number of such cosets is
    $\abs{\Z^*_{\ell/d}}/\operatorname{ord}(2, \ell/d)$. As $\abs{\Z^*_{\ell/d}} = \varphi(\ell/d)$, the claim
    follows by summing over the divisors of the number $\ell$.
  \end{proof}

  The sequence $(\operatorname{ord}(2, 2n + 1))_n$ is given as the sequence \href{https://oeis.org/A002326}{A002326} in
  the OEIS \cite{oeis}; see also \href{https://oeis.org/A037226}{A037226}.
  
  Let $n$ be an integer such that $n > 2$ and $w$ be a solution of \eqref{eq:square} of length $n$ for some values of
  $\oa$ and $\ob$ such that $11$ is not a factor of $w$ and $w$ begins with $0$. The main message of \autoref{thm:main}
  is that for each solution to \eqref{eq:square} there exists a pattern word (perhaps a trivial one) $u$ and a reversed
  standard word $v$ such that $w = \pat[v]{u}$. The word $w$ can be constructed in this manner in two ways. Indeed, if
  we denote by $L(u)$ the word we obtain from the pattern word $u$ by exchanging the letters $S$ and $L$, then clearly
  $w = \pat[L(v)]{L(u)}$. However, there is no third possibility since \autoref{prp:sl_product} ensures that $v$ is
  unique up to an application of $L$. Our conclusion is that $u$ and $v$ are unique if we insist that $w$ begins with
  $0$ and $u$ begins with $S$.

  \begin{figure}
    \centering
    \includegraphics[width=\textwidth]{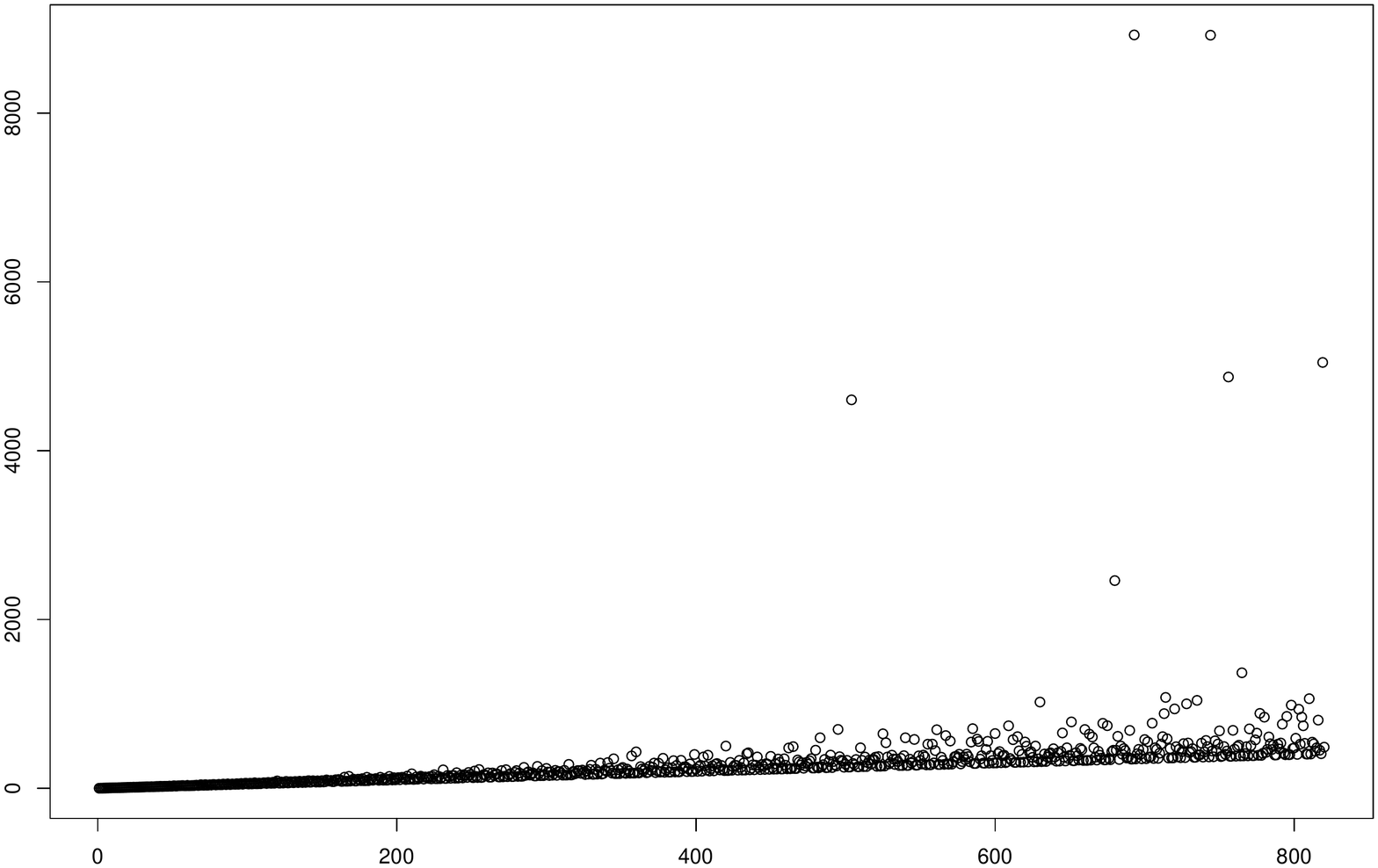}
    \caption{Plot of the number of solutions to \eqref{eq:square} of length $n$.}\label{fig:plot}
  \end{figure}

  Therefore in order to find all solutions of length $n$, we need to find all reversed standard words $v$ of length $d$
  that begin with $0$ and all pattern words $u$ of length $n/d$ that begin with $S$ for all divisors $d$ of $n$. By
  \autoref{thm:main}, we also need to require that the reversed words $v$ satisfy $\abs{v} > \abs{S_6}$ whenever the
  pattern word is nontrivial and primitive. Say $v$ is a reversed standard word such that $v$ has length $d$, $v$
  begins with $0$, and $v$ has parameters $\oa$ and $\ob$. Then $\abs{v} \leq \abs{S_6}$ if and only if
  $v = 0(10^\oa)^\ell$ for some $\ell$ such that $0 \leq \ell \leq \ob + 2$. The conclusion is that
  $\abs{v} \leq \abs{S_6}$ if and only if $v = 0$ or $\oa + 1$ and $\ell$ are divisors of $d - 1$. There is thus a
  total of $\sigma(d - 1) - 1$ reversed standard words we need to exclude ($\sigma(k)$ is the number of distinct
  divisors of $k$ with $1$ and $k$ included).
  The preceding arguments and \autoref{prp:euler} together now
  show that if $d$ is a divisor of $n$ and $d > 2$, then the contribution of reversed standard words of length $d$ to
  the total number of solutions of length $n$ is
  \begin{equation}\label{eq:phi_formula}
    \frac{\varphi(d)}{2} + \mathcal{H}(n, d),
  \end{equation}
  where
  \begin{equation}\label{eq:h_formula}
    \mathcal{H}(n, d) = (2^{\mathcal{O}(n/d) - 1} - 1) \left( \varphi(d)/2 - \sigma(d - 1) + 1 \right)
  \end{equation}
  Here the first summand of \eqref{eq:phi_formula} represents the contribution of all reversed standard words of length
  $d$ that begin with $0$ applied to the pattern word $S^{n/d}$.

  If $d = 1$ or $d = 2$, then there is only a single reversed standard word that has prefix $0$. Overall, we have that
  the number of solutions to \eqref{eq:square} of length $n$, $n > 2$, is given by
  \begin{equation*}
    \sum_{\substack{d \mid n\\2 < d \leq n}} \left( \frac{\varphi(d)}{2} + \mathcal{H}(n, d) \right) + 1 +
    \begin{cases}
        1, \text{if $n$ is even} \\
        0, \text{if $n$ is odd}.
    \end{cases}
  \end{equation*}
  By recalling the well-known identity $\sum_{d \mid n} \varphi(d) = n$, we obtain the following result.

  \begin{theorem}\label{thm:enumeration}
    The number of solutions to \eqref{eq:square} of length $n$ up to isomorphism and application of $L$ is
    \begin{equation*}
      \Floor{\frac{n}{2}} + 1 + \sum_{\substack{d \mid n\\2 < d \leq n}} \mathcal{H}(n, d)
    \end{equation*}
    where $\mathcal{H}$ is given by \eqref{eq:h_formula}.
  \end{theorem}

  Clearly the value of the formula of \autoref{thm:enumeration} is always at least $\floor{n/2} + 1$, and this lower
  bound is attained infinitely often (for example when $n$ is a prime). It seems that typically the value is close to
  $\floor{n/2}+1$ (see \autoref{fig:plot}), but for a suitable $n$ the difference can be large: the value is $1050644$
  when $n = 1736$, for example. The explanation is that $217$ is a factor of $1736$ and $\mathcal{O}(217) = 21$. Some
  additional values are given in \autoref{tbl:solutions}. The values are recorded as the sequence
  \href{https://oeis.org/A330878}{A330878} in the OEIS \cite{oeis}.

  \begin{table}
  \centering  
  \begin{tabular}{| c | c || c | c || c | c || c | c || c | c || c | c |}
    \hline
    $n$ & $S(n)$ & $n$  & $S(n)$ & $n$  & $S(n)$ & $n$  & $S(n)$ & $n$  & $S(n)$ & $n$  & $S(n)$ \\ \hline
    $1$ & $1$    & $7$  & $4$    & $13$ & $7$    & $19$ & $10$   & $25$ & $13$   & $31$ & $16$   \\
    $2$ & $2$    & $8$  & $5$    & $14$ & $8$    & $20$ & $11$   & $26$ & $14$   & $32$ & $17$   \\
    $3$ & $2$    & $9$  & $5$    & $15$ & $8$    & $21$ & $11$   & $27$ & $14$   & $33$ & $19$   \\
    $4$ & $3$    & $10$ & $6$    & $16$ & $9$    & $22$ & $12$   & $28$ & $15$   & $34$ & $18$   \\
    $5$ & $3$    & $11$ & $6$    & $17$ & $9$    & $23$ & $12$   & $29$ & $15$   & $35$ & $18$   \\
    $6$ & $4$    & $12$ & $7$    & $18$ & $10$   & $24$ & $14$   & $30$ & $16$   & $36$ & $20$   \\
    \hline
  \end{tabular}
  \caption{Number of solutions of length $n$, denoted by $S(n)$, for $n = 1, \ldots, 36$.}\label{tbl:solutions}
  \end{table}
  
  \section{Application to the Square Root Map}\label{sec:applications}
  The original reason why the specific solutions to the equation \eqref{eq:square} were studied was to construct fixed
  points of the square root map and large sets of words whose language is preserved by this mapping. In this section,
  we continue this study and apply \autoref{thm:main} to obtain a characterization of minimal subshifts whose languages
  contain arbitrarily long solutions to \eqref{eq:square} (\autoref{thm:main_square_root}).
  
  Recall that if $\infw{w}$ is an optimal squareful word with parameters $\oa$ and $\ob$ written as a product of the
  minimal squares, $\infw{w} = X_1^2 X_2^2 \dotsm$, then $\sqrt{\infw{w}} = X_1 X_2 \dotsm$. An infinite word
  $\infw{w}$ is a \emph{fixed point} if $\sqrt{\infw{w}} = \infw{w}$.

  In the paper \cite{2017:a_square_root_map_on_sturmian_words}, the first author and Whiteland showed that Sturmian
  words have the curious property that the square root map preserves their language. Alternatively phrased, we have the
  following result.

  \begin{theorem}\label{thm:sturmian_invariant}\cite[Thm.~9]{2017:a_square_root_map_on_sturmian_words}
    Let $\infw{w}$ be a Sturmian word and $\Omega$ be the Sturmian subshift generated by it. Then the subshift is
    invariant, that is, $\sqrt{\Omega} \subseteq \Omega$.
  \end{theorem}
  
  \begin{remark}
    Optimal squareful words are by definition aperiodic. Many periodic words nevertheless have a well-defined square
    root. For instance, any periodic word in $(S_5 + S_6)^\omega$ is expressible as a product of minimal squares.
    Another example is the class of so-called periodic Sturmian words. An infinite word $\infw{w}$ is a \emph{periodic
    Sturmian word} if it equals a shift of $S^\omega$ for a reversed standard word $S$.
    \autoref{thm:sturmian_invariant} is in fact true for periodic Sturmian words as well. If $\abs{S} > \abs{S_6}$,
    then the claim follows by \cite[Remark~2.4]{2020:more_on_the_dynamics_of_the_symbolic_square_root}. If
    $\abs{S} \leq \abs{S_6}$ then, strictly speaking, the parameters $\oa$ and $\ob$ might not exist as $S^\omega$ can
    contain arbitrarily large powers of $0$ or $10^\oa$, but \autoref{thm:sturmian_invariant} still holds. To see this,
    we can adapt the proof of \cite[Lemma~11]{2017:a_square_root_map_on_sturmian_words} appropriately or use a
    case-by-case analysis. For example, if $S = 10^\oa$ and $\infw{w} \in \sigma(S^\omega)$, then
    $\infw{w} = 0^\ell (10^\oa)^\omega$ for $\ell$ such that $0 \leq \ell \leq \oa$. If $\ell$ is even, then
    $\sqrt{\infw{w}} = 0^{\ell/2} \sqrt{(10^\oa)^\omega} = 0^{\ell/2}(10^\oa)^\omega$. If $\ell$ is odd, then
    $\sqrt{\infw{w}} = 0^{(\ell-1)/2} \sqrt{(010^{\oa-1})^\omega} = 0^{(\ell-1)/2}0(10^\oa)^\omega$. Thus
    $\sqrt{\infw{w}}$ has the same language as $\infw{w}$.

    Extending the results of this section to allow periodic words is mostly a matter of minor adjustments to details
    and definitions, so we focus only on the aperiodic case.
  \end{remark}

  \autoref{thm:sturmian_invariant} is somewhat unexpected as intuitively it could be expected that the square root map
  changes the language of a typical optimal squareful word. This raised the question if there are other large sets or
  subshifts than Sturmian subshifts that are invariant under the square root map. A natural way to find candidates of
  such sets is to pick a fixed point of the square root map and study the dynamics of the square root map in the
  subshift generated by this word. This is what was done in
  \cite{2017:a_square_root_map_on_sturmian_words,2020:more_on_the_dynamics_of_the_symbolic_square_root} for fixed
  points that are so-called SL-words (see below).

  How to find fixed points then? One way is to use solutions to \eqref{eq:square}. Suppose that $w$ is a solution to
  \eqref{eq:square}, that is, suppose that $\sqrt{w^2} = w$. Therefore if $(Z_n)$ is a sequence of solutions to
  \eqref{eq:square} such that its limit $\infw{w}$ has infinitely many squares $Z_n^2$ as prefixes, then
  $\sqrt{\infw{w}} = \infw{w}$. If solutions of type I are used, then $\infw{w}$ and the subshift generated by
  $\infw{w}$ are Sturmian. In
  \cite{2017:a_square_root_map_on_sturmian_words,2020:more_on_the_dynamics_of_the_symbolic_square_root}, specific type
  II solutions were considered resulting in the following theorem.

  \begin{theorem}\label{thm:sl_main_result}\cite[Thm.~2.10]{2020:more_on_the_dynamics_of_the_symbolic_square_root}
    Let $S$ be a reversed standard word with $\abs{S} > \abs{S_6}$ and $\oc$ a positive integer. Let $(Z_n)$ be the
    sequence defined by setting $Z_0 = S$ and $Z_{n+1} = L(Z_n) Z_n^{2\oc}$ for $n \geq 0$. Let $\infw{w}$ be a limit
    of $(Z_n)$ and $\Omega$ the subshift generated by $\infw{w}$. Then $\infw{w}$ is a fixed point, $\Omega$ is
    aperiodic, and for all $\infw{u} \in \Omega$ either $\sqrt{\infw{u}} \in \Omega$ or $\sqrt{\infw{u}}$ is purely
    periodic with minimum period conjugate to $S$.
  \end{theorem}

  \autoref{thm:sl_main_result} means that at least certain subshifts $\Omega$ obtained from type II solutions fail to
  be invariant. This applies even more generally; see \autoref{prp:sl_periodic}.

  \begin{definition}
    Let $S$ be a reversed standard word with $\abs{S} > \abs{S_6}$, and let $L = L(S)$. An infinite word $\infw{w}$ is
    an \emph{SL-word} if $\infw{w} \in \{S, L\}^\omega \setminus \sigma(S^\omega)$. A subshift $\Omega$ is an
    \emph{SL-subshift} if there exist fixed $S$ and $L$ such that for each $\infw{w}$ in $\Omega$ there exists an
    SL-word $\infw{u}$ in $\Omega$ such that $\infw{w}$ is a shift of $\infw{u}$.
  \end{definition}

  Notice that if $S$ is a reversed standard word, then $S$ and $L(S)$ are conjugate
  \cite[Proposition~6]{2017:a_square_root_map_on_sturmian_words}.

  \begin{proposition}\label{prp:sl_periodic}\cite[Thm.~2.11]{2020:more_on_the_dynamics_of_the_symbolic_square_root}, \cite[Lemma~48]{2017:a_square_root_map_on_sturmian_words}
    Let $\infw{w}$ be an SL-word. Then there exists a shift $\infw{u}$ of $\infw{w}$ such that $\sqrt{\infw{u}}$ is
    purely periodic with minimum period conjugate to $S$.
  \end{proposition}

  The above proposition states that if $\Omega$ is an aperiodic SL-subshift, then it cannot be invariant. The solutions
  of type II give rise to SL-subshifts, so type II solutions do not help to find non-Sturmian invariant subshifts. This
  led the authors of
  \cite{2017:a_square_root_map_on_sturmian_words,2020:more_on_the_dynamics_of_the_symbolic_square_root} to formulate
  \autoref{conjecture} stating that a minimal and invariant subshift is necessarily Sturmian. Notice that it was
  observed in \cite[Proposition~60]{2017:a_square_root_map_on_sturmian_words} that without the assumption that $\Omega$
  is minimal the claim is false.

  Using \autoref{thm:main}, we prove \autoref{conjecture} for a natural class of subshifts.

  \begin{theorem}\label{thm:main_square_root}
    Let $\Omega$ be a minimal and aperiodic subshift whose language contains infinitely many solutions to
    \eqref{eq:square}. Then either
    \begin{enumerate}[(i)]
      \item $\Omega$ is Sturmian and $\sqrt{\Omega} \subseteq \Omega$ or
      \item $\Omega$ is an SL-subshift and $\sqrt{\Omega} \not\subseteq \Omega$.
    \end{enumerate}
  \end{theorem}
  \begin{proof}
    If the language of $\Omega$ contains infinitely many reversed standard words, then there exists an infinite word
    $\infw{u}$ in $\Omega$ having arbitrarily long reversed standard words as prefixes. In other words, the word
    $\infw{u}$ is a standard Sturmian word. Then the subshift $\Omega$ is Sturmian because the subshift generated by
    $\infw{u}$ equals $\Omega$ by minimality. Then $\sqrt{\Omega} \subseteq \Omega$ by
    \autoref{thm:sturmian_invariant}. Thus we may suppose that the language of $\Omega$ contains only finitely many
    reversed standard words.

    Since the language of $\Omega$ contains infinitely many solutions to \eqref{eq:square}, we can find a sequence
    $(Z_n)$ of solutions converging to a word $\infw{w}$ in $\Omega$. Since the language of $\Omega$ contains only
    finitely many reversed standard words and $\Omega$ is aperiodic, all but finitely many solutions in the sequence
    $(Z_n)$ are powers of type II solutions by \autoref{thm:main}. There must exist a fixed reversed standard word $S$
    such that $Z_n \in \{S, L(S)\}^+ \setminus (S^+ \cup L(S)^+)$ for all $n$ large enough; otherwise there exists
    infinitely many reversed standard words in the language of $\Omega$. It follows that
    $\infw{w} \in \{S, L(S)\}^\omega \setminus \sigma(S^\omega)$, that is, the word $\infw{w}$ is an SL-word. Hence
    minimality implies that $\Omega$ is an SL-subshift. Then by \autoref{prp:sl_periodic}, there exists a word
    $\infw{v}$ in $\Omega$ such that $\sqrt{\infw{v}}$ is purely periodic. Thus $\sqrt{\infw{v}} \notin \Omega$ since
    $\Omega$ is aperiodic by assumption. In other words, we see that $\sqrt{\Omega} \not\subseteq \Omega$.
  \end{proof}

  Let us end this section by making a few remarks on attacking the remaining cases of \autoref{conjecture}. One of the
  first things that comes to mind is to ask if an invariant subshift $\Omega$ must necessarily contain a fixed point.
  So far our only method for constructing fixed points is to use solutions to \eqref{eq:square}, so if a fixed point
  exists and a fixed point must be constructed in this way, \autoref{thm:main_square_root} would show that
  \autoref{conjecture} is true. Moreover, we show in the next lemma that every square prefix of a fixed point must
  correspond to a solution to \eqref{eq:square}. However, below in \autoref{prp:no_square_prefix}, we construct an
  aperiodic fixed point having finitely many square prefixes, which casts some doubt on the workability of these ideas.

  \begin{lemma}
    Let $\infw{w}$ an optimal squareful word that is a fixed point. If $X^2$ is a prefix of $\infw{w}$, then $X$ is a
    solution to \eqref{eq:square}.
  \end{lemma}
  \begin{proof}
    This follows immediately from \autoref{lem:square_prefix}.
  \end{proof}

  \begin{proposition}\label{prp:no_square_prefix}
    There exists an optimal squareful word $\infw{w}$ that is a fixed point and has exactly one square prefix.
  \end{proposition}
  \begin{proof}
    Let $\ob = 0$, and define an infinite word $\infw{w}$ as follows:
    \begin{equation*}
      \infw{w} = S_5^2 S_6^2 \cdot \prod_{i = 0}^\infty (S_3^{2^i} S_6^{2^i})^2.
    \end{equation*}
    We show that $\infw{w}$ has the claimed properties. First of all, we have
    $S_5^2 S_6^2 (S_3 S_6)^2 \in \Pi(\oa, \ob)$, and
    \begin{equation*}
      \sqrt{S_5^2 S_6^2 S_3 S_6 S_3 S_6} = S_5 S_6 \cdot \sqrt{010^\oa 10^{\oa+1}10^\oa 010^\oa 10^{\oa+1}10^\oa} = S_5 S_6 \cdot S_2 S_1 S_6 = S_5^2 S_6^2.
    \end{equation*}
    Clearly $(S_3^{2^i} S_6^{2^i})^2 \in \Pi(\oa, \ob)$ when $i \geq 1$, and then
    \begin{equation*}
      \sqrt{(S_3^{2^i} S_6^{2^i})^2} = (S_3^{2^{i-1}} S_6^{2^{i-1}})^2.
    \end{equation*}
    Therefore $\infw{w}$ is optimal squareful and $\sqrt{\infw{w}} = \infw{w}$. Let $u^2$ be a prefix of $\infw{w}$
    with $\abs{u} \geq 2\abs{S_5 S_6 S_3 S_6}$. In particular, the word $S_5^2 S_6^2 0$ is a prefix of $u$. The suffix
    $u$ of $u^2$ occurs in $\infw{w}$ in some concatenation
    $\smash[t]{(S_3^{2^i} S_6^{2^i})^2 (S_3^{2^{i+1}} S_6^{2^{i+1}})^2}$ of two blocks for some $i \geq 1$. Let
    $z = (10^{\oa+1})^3 10^\oa$. Now $u$ has prefix $z$. No word of the form $\smash[t]{S_6^{2j}}$ with $j \geq 1$
    contains $(10^{\oa+1})^2 1$, so either $z$ is a suffix of $\smash[t]{S_3^{2^i}}$ or $\smash[t]{S_3^{2^{i+1}}}$ or
    the prefix $\smash[t]{10^{\oa+1}10^\oa}$ of $z$ is a suffix of $\smash[t]{S_6^{2^i}}$. The prefix $z$ is followed in $u$ by $S_6 0$.
    Since $S_6 0$ is not a prefix of
    $S_6^2$, it follows that the latter option is true: the prefix $10^{\oa+1}10^\oa$ of $z$ is a suffix of
    $\smash[t]{S_6^{2^i}}$. This in turn means that $\smash[t]{S_6^{2^i}}$ is followed by $010^{\oa+1} 10^\oa 10^{\oa+1} 10^{\oa+1}$. This
    word is not a prefix of $S_3^2 S_6^2$ or $S_3^4$, so we obtain a contradiction. The conclusion is that $u$ does not
    exist, and for each square prefix $v^2$ of $\infw{w}$, we have $\abs{v} < 2\abs{S_5 S_6 S_3 S_6}$. It is a
    straightforward task to verify that $\infw{w}$ has exactly one square prefix.
  \end{proof}

  Notice that the word $\infw{w}$ of the proof of \autoref{prp:no_square_prefix} is not a counterexample to
  \autoref{conjecture}. First of all, the subshift generated by $\infw{w}$ is not minimal. Secondly, the subshift
  generated by $\Omega$ is not invariant. Indeed, consider the square root of the word $\infw{z}$ obtained from
  $\infw{w}$ by removing its first $\abs{S_6}$ letters. We have
  \begin{equation*}
    \sqrt{\infw{z}} = 010^\oa 10^{\oa+1} 10^{\oa+1} 10^\oa 10^{\oa+1} 10^{\oa+1} 10^\oa 10^{\oa+1} 10^{\oa+1} 10^\oa 1 \dotsm.
  \end{equation*}
  Similar to the proof of \autoref{prp:no_square_prefix}, it is straightforward to verify that this visible prefix is
  not a factor of $\infw{w}$.

  We considered above only square prefixes, but arbitrary positions could be considered as well. While an optimal
  squareful word contains an occurrence of a square at each position, it is unclear if long squares need to appear.
  Every position of a Sturmian word begins with arbitrarily long squares
  \cite[Lemma~8]{2001:transcendence_of_sturmian_or_morphic_continued_fractions}, and the language of the word
  $\infw{w}$ constructed in the proof of \autoref{prp:no_square_prefix} contains infinitely many squares.

  \begin{question}
    Does there exist an optimal squareful word whose language contains only finitely many squares?
  \end{question}

  \section{Remark on Periodic Points}\label{sec:remarks}
  The results presented in this paper concern fixed points of the square root map, but other periodic points could be
  considered as well. Clearly solutions to \eqref{eq:square} are not helpful in constructing $p$-periodic points when
  $p > 1$. We are not aware of a general construction method, but we show next that aperiodic proper $2$-periodic
  points exist. By $\sqrt[2]{\infw{w}}$ we mean the second iterate of $\infw{w}$.

  \begin{proposition}\label{prp:2-periodic_point}
    There exists an optimal squareful word $\infw{w}$ such that $\sqrt{\infw{w}} \neq \infw{w}$ and
    $\sqrt[2]{\infw{w}} = \infw{w}$.
  \end{proposition}
  \begin{proof}
    Define two integer sequences $(r(n))_n$ and $(s(n))_n$ as follows: $r(1) = s(1) = 2$, $r(2) = 6$, $s(2) = 8$ and
    $r(n+1) = 4r(n)$, $s(n+1) = 4s(n)$ for all $n \geq 2$. Let $\ob = 0$, and define
    \begin{equation*}
      \infw{w} = S_2^2 S_1^2 \cdot \prod_{n=1}^\infty (S_6^2)^{r(n)} (S_3^2)^{s(n)},
    \end{equation*}
    so that
    \begin{equation*}
      \sqrt{\infw{w}} = S_2 S_1 \cdot \prod_{n=1}^\infty S_6^{r(n)} S_3^{s(n)} = S_2 S_1 S_6^2 S_3^2 \cdot \prod_{n=2}^\infty S_6^{r(n)} S_3^{s(n)} = S_2^2 S_1^2 S_4^2 S_3^2 S_3 \cdot \prod_{n=2}^\infty S_6^{r(n)} S_3^{s(n)}.
    \end{equation*}
    Since $S_2^2 S_1^2 S_4^2$ is not a prefix of $\infw{w}$, we see that $\sqrt{\infw{w}} \neq \infw{w}$. Notice that
    the word $S_2 S_1 S_4$ is a solution to \eqref{eq:square}. Since $S_3 S_6^2 = S_2^2 S_1^2 S_4^2 S_3$, we see that
    $S_3 (S_6^2)^i = (S_2^2 S_1^2 S_4^2)^i S_3 = (S_2 S_1 S_4)^{2i} S_3$ for all $i$. It follows that
    \begin{equation*}
      S_2^2 S_1^2 S_4^2 S_3^2 S_3 \cdot S_6^6 S_3^8 = S_2^2 S_1^2 S_4^2 S_3^2 \cdot (S_2^2 S_1^2 S_4^2)^3 S_3 \cdot S_3^8,
    \end{equation*}
    so
    \begin{align*}
      \sqrt[2]{\infw{w}} &= S_2 S_1 S_4 S_3 (S_2 S_1 S_4)^3 S_3^4 \cdot \sqrt{\prod_{n=3}^\infty S_3 S_6^{r(n)} S_3^{s(n)-1}} \\
                         &= S_2^2 S_1^2 (S_6^2)^2 S_3^3 \cdot \sqrt{\prod_{n=3}^\infty (S_2^2 S_1^2 S_4^2)^{r(n)/2} S_3^{s(n)}} \\
                         &= S_2^2 S_1^2 (S_6^2)^2 S_3^3 \cdot \prod_{n=3}^\infty (S_2 S_1 S_4)^{r(n)/2} S_3^{s(n)/2} \\
                         &= S_2^2 S_1^2 (S_6^2)^2 S_3^3 \cdot \prod_{n=3}^\infty (S_2^2 S_1^2 S_4^2)^{r(n)/4} (S_3^2)^{s(n)/4} \\
                         &= S_2^2 S_1^2 (S_6^2)^2 S_3^4 \cdot \prod_{n=3}^\infty (S_6^2)^{r(n)/4} (S_3^2)^{s(n)/4} \\
                         &= S_2^2 S_1^2 (S_6^2)^2 (S_3^2)^2 \cdot \prod_{n=2}^\infty (S_6^2)^{r(n)} (S_3^2)^{s(n)} \\
                         &= \infw{w}.
    \end{align*}
    Since the sequences $(r(n))_n$ and $(s(n))_n$ are increasing, the word $\infw{w}$ is aperiodic. It is
    straightforward to check that $\infw{w}$ is optimal squareful.
  \end{proof}

  We do not how to produce $p$-periodic points for $p > 2$. Moreover, we are not completely satisfied with the word
  $\infw{w}$ of the proof of \autoref{prp:2-periodic_point} because the subshift generated by $\infw{w}$ is not
  minimal, that is, the word $\infw{w}$ is not uniformly recurrent.

  \begin{question}
    Does there exist an optimal squareful word that is a proper $p$-periodic point for all $p > 1$? Can such words be
    taken as uniformly recurrent?
  \end{question}

  \printbibliography
	
\end{document}